\documentclass[aps,preprint,letterpaper]{revtex4}

\usepackage{amsfonts,amsmath,amssymb}
\usepackage{graphicx}
\usepackage{bm,bbm}
\usepackage[utf8]{inputenc}
\usepackage{url}

\newtheorem{definition}{Definition}[section]
\newtheorem{theorem}[definition]{Theorem}
\newenvironment{proof}[1][Proof]{\noindent\textbf{#1.} }{\ \rule{0.5em}{0.5em}}

\newcommand{\g}{\ensuremath{\mathfrak{g}}}
\newcommand{\G}{\ensuremath{\mathfrak{G}}}

\newcommand{\ucsc}{Departamento de Matemática y Física Aplicadas, Universidad Católica de la Santísima Concepción, Alonso de Ribera 2850, Concepción, Chile}
\newcommand{\udec}{Departamento de Física, Universidad de Concepción, Casilla 160-C, Concepción, Chile}
\newcommand{\mpi}{Max-Planck-Institut für Gravitationsphysik, Albert Einstein Institute. Am Mühlenberg~1, D-14476 Golm bei Potsdam, Germany}

\begin{document}

\title{Dual Formulation of the Lie Algebra $S$-expansion Procedure}

\author{Fernando Izaurieta}
\email{fizaurie@ucsc.cl}
\affiliation{\ucsc}

\author{Alfredo Pérez}
\email{alfperez@udec.cl}
\affiliation{\udec}
\affiliation{\mpi}

\author{Eduardo Rodríguez}
\email{edurodriguez@ucsc.cl}
\affiliation{\ucsc}

\author{Patricio Salgado}
\email{pasalgad@udec.cl}
\affiliation{\udec}

\date{\today}

\begin{abstract}
The \emph{expansion} of a Lie algebra entails finding a new, bigger algebra \G, through a series of well-defined steps, from an original Lie algebra \g. One incarnation of the method, the so-called $S$-expansion, involves the use of a finite abelian semigroup $S$ to accomplish this task. In this paper we put forward a dual formulation of the $S$-expansion method which is based on the dual picture of a Lie algebra given by the Maurer--Cartan forms. The dual version of the method is useful in finding a generalization to the case of a gauge free differential algebra, which in turn is relevant for physical applications in, e.g., Supergravity. It also sheds new light on the puzzling relation between two Chern--Simons Lagrangians for gravity in $2+1$ dimensions, namely the Einstein--Hilbert Lagrangian and the one for the so-called `exotic gravity.'
\end{abstract}




\maketitle

\section{Introduction}

Lie algebra expansions were introduced by the first time in Ref.~\cite{Hat01}, and the method was subsequently studied in general in Refs.~\cite{deAz02,deAz07,deAz04}.
The idea is to perform a rescaling by a parameter $\lambda$ of some of the group coordinates $g^{i}, i=1,\ldots, \dim \g$. Consequently, the Maurer--Cartan (MC) one-forms $\omega^{i} \left( g, \lambda \right)$ of \g\ are expanded as a power series in $\lambda$.
Inserting these expansions back in the original MC equations for \g, one obtains the MC equations of a new finite-dimensional expanded Lie algebra.

An alternative expansion procedure to the method of power series expansion is the $S$-expansion method~\cite{Iza06b,Iza08b}.
The $S$-expansion method allows us also to obtain new Lie algebras starting from an original one by choosing an Abelian semigroup $S$ and applying the general theorems~4.2 and~6.1 from Ref.~\cite{Iza06b}, which give us, following the terminology introduced in Ref.~\cite{Iza06b}, ``resonant subalgebras'' and what has been dubbed ``reduced algebras.''

In spite of the examples given in Ref.~\cite{Iza06b}, the relation between both procedures has remained as an interesting open problem basically because
(i)~the $S$-expansion is defined as the action of a semigroup $S$ on the generators $T_{A}$ of the algebra and the power series expansion is carried out on the MC forms of the original algebra, and
(ii)~the $S$-expansion is defined on the algebra \g\ without referring to the group manifold, whereas the power series expansion is based on a rescaling of the group coordinates.

It is the purpose of this paper to study the $S$-expansion procedure in the context of the group manifold and then to find the dual formulation of such $S$-expansion procedure.

The article is organized as follows.
In section~\ref{ssexpa} we review the main aspects of the $S$-expansion procedure.
In section~\ref{sdual} we shall construct the dual formulation of the Lie algebra $S$-expansion procedure.
The expansion of free differential algebras is considered in section~\ref{sfree}.
In section~\ref{scs} several applications are shown, e.g., how to obtain $(2+1)$-dimensional Chern--Simons (CS) AdS gravity from the so-called exotic gravity.
Section~\ref{sfin} concludes the work with an outlook to further applications in gravity and supergravity.

\section{\label{ssexpa}$S$-expansion of Lie Algebras}

In this section we shall review the main aspects of the $S$-expansion procedure introduced in Ref.~\cite{Iza06b}.
The $S$-expansion method is based on combining the structure constants of a Lie algebra \g\ with the inner multiplication law of a semigroup $S$ to define the Lie bracket of a new, $S$-expanded algebra $\G = S \times \g$.

Let $S = \left\{ \lambda_{\alpha}, \alpha = 1, \ldots, N \right\}$
be a finite abelian semigroup with ``two-selector'' $K_{\alpha \beta}^{\phantom{\alpha \beta} \gamma}$ defined by
\begin{equation}
K_{\alpha \beta}^{\phantom{\alpha \beta} \gamma} = \left\{
 \begin{array}{cl}
  1, & \text{when } \lambda_{\alpha} \lambda_{\beta} = \lambda_{\gamma} \\
  0, & \text{otherwise},
 \end{array}
\right.
\end{equation}
and let \g\ be a Lie algebra with structure constants $C_{AB}^{\phantom{AB}C}$,
\begin{equation}
\left[ \bm{T}_{A}, \bm{T}_{B} \right] = C_{AB}^{\phantom{AB}C} \bm{T}_{C}.
\end{equation}
Then it may be shown~\cite{Iza06b} that the product $\G = S \times \g$ corresponds to the Lie algebra given by
\begin{equation}
\left[ \bm{T}_{\left( A, \alpha \right)}, \bm{T}_{\left( B, \beta \right)} \right] = 
K_{\alpha \beta}^{\phantom{\alpha \beta} \gamma} C_{AB}^{\phantom{AB}C} \bm{T}_{\left( C, \gamma \right)}.
\label{s2'}
\end{equation}

\begin{theorem}
The product $\left[ \cdot , \cdot \right]$ defined in eq.~(\ref{s2'}) is also a Lie
product because it is linear, antisymmetric and satisfies the Jacobi identity. This
product defines a new Lie algebra characterized by the pair
$\left( \G, \left[ \cdot , \cdot \right] \right)$,
which is called \emph{$S$-expanded Lie algebra}.
\label{tsexp}
\end{theorem}
\begin{proof}
The proof is direct and may be found in Ref.~\cite{Iza06b}.
\end{proof}

\section{\label{sdual}Dual Formulation of the $S$-expansion Procedure}

Theorem~\ref{tsexp} implies that, for every abelian semigroup $S$ and Lie algebra \g, the product $\G = S \times \g$ is also a Lie algebra, with a Lie bracket given by eq.~(\ref{s2'}). This in turn means that it must be possible to look at this $S$-expanded Lie algebra \G\ from the dual point of view of the MC forms.

\begin{theorem}
Let $S = \left\{ \lambda_{\alpha}, \alpha = 1, \ldots, N \right\}$
be a finite abelian semigroup and let $\omega^{A}$ be the MC forms for a Lie algebra \g.
Then, the MC forms $\omega^{\left( A, \alpha \right)}$ associated with the $S$-expanded Lie algebra $\G = S \times \g$ [cf. Theorem~\ref{tsexp}] are related to the $\omega^{A}$ by
\begin{equation}
\omega^{A} = \lambda_{\alpha} \omega^{\left( A, \alpha \right)}.
\label{cinco}
\end{equation}
By definition, these forms satisfy the MC equations
\begin{equation}
\mathrm{d} \omega^{\left( C, \gamma \right)} + \frac{1}{2} K_{\alpha \beta}^{\phantom{\alpha \beta} \gamma} C_{AB}^{\phantom{AB}C}
\omega^{\left( A, \alpha \right)} \omega^{\left( B, \beta \right)} = 0.
\label{seis}
\end{equation}
\end{theorem}

\begin{proof}
The simplest way to check the validity of eq.~(\ref{cinco}) is by multiplying eq.~(\ref{seis}) by $\lambda_{\gamma}$ and using the defining relation for the two-selector $K_{\alpha \beta}^{\phantom{\alpha \beta} \gamma}$, namely
$\lambda_{\alpha} \lambda_{\beta} = K_{\alpha \beta}^{\phantom{\alpha \beta} \gamma} \lambda_{\gamma}$. We get
\begin{align}
\lambda_{\gamma} \left[ \mathrm{d} \omega^{\left( C, \gamma \right)} + \frac{1}{2} K_{\alpha \beta}^{\phantom{\alpha \beta} \gamma} C_{AB}^{\phantom{AB}C}
\omega^{\left( A, \alpha \right)} \omega^{\left( B, \beta \right)} \right] & = 0 \\
 \mathrm{d} \left[ \lambda_{\gamma} \omega^{\left( C, \gamma \right)} \right] + \frac{1}{2} C_{AB}^{\phantom{AB}C}
\left[ \lambda_{\alpha} \omega^{\left( A, \alpha \right)} \right] \left[ \lambda_{\beta} \omega^{\left( B, \beta \right)} \right] & = 0.
\label{epro}
\end{align}
The required identification is obtained by matching eq.~(\ref{epro}) with the MC equations for \g.
This concludes the proof.
\end{proof}

It is perhaps interesting to notice that the relation shown in eq.~(\ref{cinco}) is analogous to the method of power series expansion developed in Ref.~\cite{deAz02}.

\subsection{$0_{S}$-Reduction of $S$-expanded Lie Algebras}

The concept of \emph{reduction} of Lie algebras, and in particular $0_{S}$-reduction, was introduced in Ref.~\cite{Iza06b}. As implied by the name, it involves the extraction of a smaller algebra from a given Lie algebra \g\ when certain conditions are met. In spite of the superficial similarity of the concepts, a reduced algebra is not, in general, a subalgebra of \g~\cite{Iza06b}.

In this section we present the dual formulation for the $0_{S}$-reduction of an $S$-expanded Lie algebra \G, formulated in the language of the MC forms.

Let $S = \left \{ \lambda_{i}, i = 1, \ldots, N\right\} \cup \left\{ \lambda_{N+1} = 0_{S} \right \}$
be an abelian semigroup with zero.
The expanded MC forms $\omega^{\left( A, \alpha \right)}$ are then given by
\begin{equation}
\omega^{A} = \lambda_{i} \omega^{\left( A, i \right)} +
0_{S} \tilde{\omega}^{A},
\label{diecisiete}
\end{equation}
where $\tilde{\omega}^{A} = \omega^{\left( A, N+1 \right)}$.
We shall show that the MC forms $\omega^{\left( A, i \right)}$ by themselves (without including $\tilde{\omega}^{A}$) are those of a Lie algebra---the $0_{S}$-reduced algebra $\G_{R}$.

It can be shown~\cite{Iza06b} that $K_{ij}^{\phantom{ij}k} C_{AB}^{\phantom{AB}C}$ are the structure constants for the $0_{S}$-reduced $S$-expanded algebra $\G_{R}$, which is generated by $\bm{T}_{\left( A, i \right)}$:
\begin{equation}
\left[ \bm{T}_{\left( A, i \right)}, \bm{T}_{\left( B, j \right)} \right] =
K_{ij}^{\phantom{ij}k} C_{AB}^{\phantom{AB}C} \bm{T}_{\left( C, k \right)}.
\end{equation}
Theorem~\ref{t0redLie} below gives the equivalent statement in terms of MC forms.

\begin{theorem}
Let $S = \left \{ \lambda_{i}, i = 1, \ldots, N\right\} \cup \left\{ \lambda_{N+1} = 0_{S} \right \}$
be an abelian semigroup with zero and let $\left \{ \omega^{\left( A, i \right)}, i = 1, \ldots, N\right\} \cup \left\{ \omega^{\left( A, N+1 \right)} = \tilde{\omega}^{A} \right\}$ be the MC forms for the $S$-expanded algebra
$\G = S \times \g$ of \g\ by the semigroup $S$. Then, $\left \{ \omega^{\left( A, i \right)}, i = 1, \ldots, N\right\}$ are the MC forms for the $0_{S}$-reduced $S$-expanded algebra $\G_{R}$.
\label{t0redLie}
\end{theorem}

\begin{proof}
The MC forms for the $S$-expanded algebra \G\ satisfy the MC equations [cf.~eq.~(\ref{seis})]
\begin{equation}
\mathrm{d} \omega^{\left( C, \gamma \right)} + \frac{1}{2} K_{\alpha \beta}^{\phantom{\alpha \beta} \gamma} C_{AB}^{\phantom{AB}C}
\omega^{\left( A, \alpha \right)} \omega^{\left( B, \beta \right)} = 0.
\end{equation}
The $\gamma = k$ component reads
\begin{equation}
\mathrm{d} \omega^{\left( C, k \right)} + \frac{1}{2} K_{\alpha \beta}^{\phantom{\alpha \beta} k} C_{AB}^{\phantom{AB}C}
\omega^{\left( A, \alpha \right)} \omega^{\left( B, \beta \right)} = 0.
\end{equation}
Summing over $\alpha$ and $\beta$ and noting that
$K_{i,N+1}^{\phantom{i,N+1} k} = K_{N+1,j}^{\phantom{N+1,j} k} = K_{N+1,N+1}^{\phantom{N+1,N+1} k} = 0$
we get
\begin{equation}
\mathrm{d} \omega^{\left( C, k \right)} + \frac{1}{2} K_{ij}^{\phantom{ij} k} C_{AB}^{\phantom{AB}C}
\omega^{\left( A, i \right)} \omega^{\left( B, j \right)} = 0,
\end{equation}
which shows that $\left \{ \omega^{\left( A, i \right)}, i = 1, \ldots, N\right\}$ are the MC forms for a Lie algebra whose structure constants are $K_{ij}^{\phantom{ij}k} C_{AB}^{\phantom{AB}C}$, as we set out to prove.
\end{proof}

\section{\label{sfree}Generalization to the Case of a Gauge Free Differential Algebra}

A free differential algebra (FDA) is a set of differential forms closed under the action of the exterior product%
~\footnote{The wedge symbol $\wedge$ is omitted throughout this work, although wedge product is always assumed between forms.}
$\wedge$ and the exterior derivative $\mathrm{d}$ \cite{Sul77,DAu82}.

Every Lie algebra \g\ leads to a `gauge' FDA through its MC equations.
To obtain a gauge FDA we replace the MC forms $\omega^{A}$ by the gauge field one-forms $A^{A}$ and introduce their corresponding curvatures $F^{A}$ by
\begin{equation}
F^{A} = \mathrm{d} A^{A} + \frac{1}{2} C_{BC}^{\phantom{BC}A} A^{B} A^{C}.
\label{trece}
\end{equation}
This equation expressing the exterior derivative of the gauge connection $A^{A}$ in terms of the curvature $F^{A}$ must be supplemented with the Bianchi identity, which expresses the exterior derivative of $F^{A}$ in terms of $A^{A}$ and $F^{A}$ itself:
\begin{equation}
\mathrm{d} F^{A} + C_{BC}^{\phantom{BC}A} A^{B} F^{C} = 0.
\label{ebianchi}
\end{equation}
Eqs.~(\ref{trece})--(\ref{ebianchi}) form the gauge FDA on which we shall perform the $S$-expansion.

\begin{theorem}
Let $S = \left\{ \lambda_{\alpha}, \alpha = 1, \ldots, N \right\}$
be a finite abelian semigroup and let \g\ be a Lie algebra.
Then, the `expanded' connection $A^{\left( A, \alpha \right)}$ and curvature $F^{\left( A, \alpha \right)}$ defined by [see eq.~(\ref{cinco})]
\begin{align}
A^{A} & = \lambda_{\alpha} A^{\left( A, \alpha \right)}, \label{catorce} \\
F^{A} & = \lambda_{\alpha} F^{\left( A, \alpha \right)}, \label{eFexp}
\end{align}
form a gauge FDA for the $S$-expanded Lie algebra $\G = S \times \g$, with the defining equations
\begin{align}
\mathrm{d} A^{\left( A, \alpha \right)} + \frac{1}{2}
K_{\beta \gamma}^{\phantom{\beta \gamma}\alpha} C_{BC}^{\phantom{BC}A}
A^{\left( B, \beta \right)} A^{\left( C, \gamma \right)}
& = F^{\left( A, \alpha \right)}, \label{eFexpdef}\\
\mathrm{d} F^{\left( A, \alpha \right)} +
K_{\beta \gamma}^{\phantom{\beta \gamma}\alpha} C_{BC}^{\phantom{BC}A}
A^{\left( B, \beta \right)} F^{\left( C, \gamma \right)} & = 0. \label{eDFexp}
\end{align}
\end{theorem}

\begin{proof}
We must show that the expanded connection and curvature defined in eqs.~(\ref{catorce})--(\ref{eFexp}) satisfy equations analogous to~(\ref{trece})--(\ref{ebianchi}), with the structure constants of \g, $C_{AB}^{\phantom{AB}C}$, replaced by the structure constants of \G, $K_{\alpha \beta}^{\phantom{\alpha \beta} k} C_{AB}^{\phantom{AB}C}$ [cf.~eqs.(\ref{eFexpdef})--(\ref{eDFexp})]. Let us start with eq.~(\ref{trece}): substituting eqs.~(\ref{catorce})--(\ref{eFexp}) into eq.~(\ref{trece}) we get
\begin{align}
\lambda_{\alpha} F^{\left( A, \alpha \right)} & =
\mathrm{d} \left[ \lambda_{\alpha} A^{\left( A, \alpha \right)} \right]
+ \frac{1}{2} C_{BC}^{\phantom{BC}A}
\left[ \lambda_{\beta}  A^{\left( B, \beta  \right)} \right]
\left[ \lambda_{\gamma} A^{\left( C, \gamma \right)} \right]
\nonumber \\ & =
\lambda_{\alpha} \left[ \mathrm{d} A^{\left( A, \alpha \right)} + \frac{1}{2}
K_{\beta \gamma}^{\phantom{\beta \gamma}\alpha} C_{BC}^{\phantom{BC}A}
A^{\left( B, \beta \right)} A^{\left( C, \gamma \right)} \right].
\label{quince}
\end{align}
Equating coefficients on both sides we readily recover eq.~(\ref{trece}).
Replacing now eqs.~(\ref{catorce})--(\ref{eFexp}) into eq.~(\ref{ebianchi}) we find
\begin{align}
\mathrm{d} \left[ \lambda_{\alpha} F^{\left( A, \alpha \right)} \right] +
C_{BC}^{\phantom{BC}A}
\left[ \lambda_{\beta}  A^{\left( B, \beta  \right)} \right]
\left[ \lambda_{\gamma} F^{\left( C, \gamma \right)} \right] & = 0
\nonumber \\
\lambda_{\alpha} \left[ \mathrm{d} F^{\left( A, \alpha \right)} +
K_{\beta \gamma}^{\phantom{\beta \gamma}\alpha} C_{BC}^{\phantom{BC}A}
A^{\left( B, \beta \right)} F^{\left( C, \gamma \right)} \right] & = 0,
\end{align}
which finishes the proof.
\end{proof}

\subsection{$0_{\text{S}}$-Reduction of Free Differential Algebras}

Let $S = \left \{ \lambda_{i}, i = 1, \ldots, N\right\} \cup \left\{ \lambda_{N+1} = 0_{S} \right \}$
be an abelian semigroup with zero.
The expanded connection and curvature $A^{\left( A, \alpha \right)}$, $F^{\left( A, \alpha \right)}$ are then given by
\begin{align}
A^{A} = \lambda_{i} A^{\left( A, i \right)} + 0_{S} \tilde{A}^{A}, \\
F^{A} = \lambda_{i} F^{\left( A, i \right)} + 0_{S} \tilde{F}^{A},
\end{align}
where $\tilde{A}^{A} = A^{\left( A, N+1 \right)}$, $\tilde{F}^{A} = F^{\left( A, N+1 \right)}$.
We shall show that the expanded forms $A^{\left( A, i \right)}$, $F^{\left( A, i \right)}$ by themselves (without including either $\tilde{A}^{A}$ nor $\tilde{F}^{A}$) are those of a gauge FDA---a $0_{S}$-reduced gauge FDA.

\begin{theorem}
Let $S = \left \{ \lambda_{i}, i = 1, \ldots, N\right\} \cup \left\{ \lambda_{N+1} = 0_{S} \right \}$
be an abelian semigroup with zero and let
$\left \{ A^{\left( A, i \right)}, i = 1, \ldots, N\right\} \cup \left\{ A^{\left( A, N+1 \right)} = \tilde{A}^{A} \right\}$,
$\left \{ F^{\left( A, i \right)}, i = 1, \ldots, N\right\} \cup \left\{ F^{\left( A, N+1 \right)} = \tilde{F}^{A} \right\}$
be the connection and curvature for the $S$-expanded gauge FDA obtained by expanding $\left( A^{A}, F^{A} \right)$ by the semigroup $S$. Then,
$\left \{ A^{\left( A, i \right)}, i = 1, \ldots, N\right\}$,
$\left \{ F^{\left( A, i \right)}, i = 1, \ldots, N\right\}$,
are the connection and curvature for a new, $0_{S}$-reduced gauge FDA, with the defining equations
\begin{align}
\mathrm{d} A^{\left( A, i \right)} + \frac{1}{2}
K_{j k}^{\phantom{j k}i} C_{BC}^{\phantom{BC}A}
A^{\left( B, j \right)} A^{\left( C, k \right)}
& = F^{\left( A, i \right)}, \label{e0redA} \\
\mathrm{d} F^{\left( A, i \right)} +
K_{j k}^{\phantom{j k}i} C_{BC}^{\phantom{BC}A}
A^{\left( B, j \right)} F^{\left( C, k \right)} & = 0. \label{e0redF}
\end{align}
\label{t0redfda}
\end{theorem}

\begin{proof}
Since $A^{\left( A, \alpha \right)}$ and $F^{\left( A, \alpha \right)}$ are the connection and curvature for an $S$-expanded FDA, they satisfy the defining eqs.~(\ref{eFexpdef})--(\ref{eDFexp}). To prove eqs.~(\ref{e0redA})--(\ref{e0redF}) it suffices to take the $\alpha = i$ component in eqs.~(\ref{eFexpdef})--(\ref{eDFexp}) and sum over $\beta$ and $\gamma$, taking into account the fact that
$K_{j,N+1}^{\phantom{j,N+1} i} = K_{N+1,k}^{\phantom{N+1,k} i} = K_{N+1,N+1}^{\phantom{N+1,N+1} i} = 0$.
The interested reader will find that the details of the calculation are easily filled in. This concludes the proof.
\end{proof}

\section{\label{scs}Applications to Chern--Simons Theory of Gravity}

Our main interest in developing methods for expanding Lie and gauge free differential algebras lies on their expected applications in gravity and supergravity.

\subsection{Three-dimensional Gravity Revisited}

As already noted by Witten in his classic 1988 paper~\cite{Wit88}, the Einstein field equations for General Relativity in three-dimensional space-time can be derived from the `exotic' Lagrangian
\begin{equation}
L_{\text{ex}} = \omega_{\phantom{a}b}^{a} \mathrm{d} \omega_{\phantom{b}a}^{b} +
\frac{2}{3} \omega_{\phantom{a}b}^{a} \omega_{\phantom{b}c}^{b} \omega_{\phantom{c}a}^{c},
\label{eexo}
\end{equation}
which is a function of the spin connection $\omega^{ab}$ only [the vielbein $e^{a}$ is conspicuously absent from~(\ref{eexo})].

The Lagrangian~(\ref{eexo}) can be written as a CS form for the $\left( 2+1 \right)$-dimensional Lorentz algebra $\mathfrak{L}$,
\begin{equation}
\left[ \bm{J}_{ab}, \bm{J}_{cd} \right] = \eta_{cb} \bm{J}_{ad} - \eta_{ca} \bm{J}_{bd} +
\eta_{db} \bm{J}_{ca} - \eta_{da} \bm{J}_{cb}.
\end{equation}
The $\mathfrak{L}$-valued, one-form gauge connection $\bm{A}$ and two-form curvature $\bm{F}$ read
\begin{align}
\bm{A} & = \frac{1}{2} \omega^{ab} \bm{J}_{ab}, \\
\bm{F} & = \frac{1}{2} R^{ab} \bm{J}_{ab},
\label{29}
\end{align}
where the Lorentz curvature $R^{ab}$ is defined as usual, $R^{ab} = \mathrm{d} \omega^{ab} + \omega_{\phantom{a}c}^{a} \omega^{cb}$.

A straightforward calculation shows that
\begin{equation}
L_{\text{ex}} = 2 \left\langle \bm{A} \mathrm{d} \bm{A} +
\frac{2}{3} \bm{A}^{3} \right\rangle,
\end{equation}
where $\left\langle \cdots \right\rangle$ stands for the following rank two, symmetric invariant tensor:
\begin{equation}
\left\langle \bm{J}_{ab} \bm{J}_{cd} \right\rangle =
\eta_{ad} \eta_{bc} - \eta_{ac} \eta_{bd}.
\label{einvtenL}
\end{equation}

By contrast, the usual Einstein--Hilbert (HE) Lagrangian (with a negative cosmological constant $\Lambda = -3/\ell^2$, where $\ell$ is a length),
\begin{equation}
L_{\text{EH}} = \frac{1}{\ell} \varepsilon_{abc}
\left( R^{ab} + \frac{1}{3\ell^{2}} e^{a} e^{b} \right) e^{c},
\end{equation}
requires that we consider the full AdS algebra $\mathfrak{so} \left( 2,2 \right)$, with the one-form gauge connection
\begin{equation}
\bm{A} = \frac{1}{\ell} e^{a} \bm{P}_{a} + \frac{1}{2} \omega^{ab} \bm{J}_{ab}.
\end{equation}

In this section we cast the relationship between~(\ref{eexo}) and the usual Einstein--Hilbert Lagrangian (with a cosmological constant) in terms of an $S$-expansion of the relevant FDA.

In order to ease the comparison, it proves useful to define%
~\footnote{Here $\varepsilon$ is the Levi-Civita symbol, with $\varepsilon^{012} = \varepsilon_{012} = +1$.}
\begin{align}
\bm{J}^{a} & = \frac{1}{2} \varepsilon^{abc} \bm{J}_{bc}, \\
\omega_{a} & = \frac{1}{2} \varepsilon_{abc} \omega^{bc},
\end{align}
so that $\bm{A}$ and $\bm{F}$ now take the form
\begin{align}
\bm{A} & = \omega_{a} \bm{J}^{a}, \\
\bm{F} & = F_{a} \bm{J}^{a},
\end{align}
with
\begin{equation}
F_{a} = \frac{1}{2} \varepsilon_{abc} R^{bc} =
\mathrm{d} \omega_{a} - \frac{1}{2} \eta_{ab} \varepsilon^{bcd} \omega_{c} \omega_{d}.
\label{eFa}
\end{equation}

From this $F_{a}$ we can construct the following invariant polynomial:
\begin{equation}
P = F^{a} F_{a} = \mathrm{d} L_{\text{CS}},
\end{equation}
which shows that~(\ref{eexo}) is a CS form, quasi-invariant under Lorentz transformations.

\subsection{The $\mathbbm{Z}_{2}$-expansion}

Let us consider the (semi)group
$\mathbbm{Z}_{2} = \left\{ \lambda_{0}, \lambda_{1} \right\}$,
provided with the following product:
\begin{align}
\lambda_{0} \lambda_{0} & = \lambda_{0}, \label{e00} \\
\lambda_{0} \lambda_{1} & = \lambda_{1}, \\
\lambda_{1} \lambda_{0} & = \lambda_{1}, \\
\lambda_{1} \lambda_{1} & = \lambda_{0}. \label{e11}
\end{align}

Following the ideas from section~\ref{sfree}, we define the $S$-expanded spin connection and curvature by
\begin{align}
\omega_{a} & = \lambda_{0} \omega_{a}^{\left( 0 \right)} + \lambda_{1} \omega_{a}^{\left( 1 \right)},
\label{treinta} \\
F_{a} & = \lambda_{0} F_{a}^{\left( 0 \right)} + \lambda_{1} F_{a}^{\left( 1 \right)}.
\label{treintayuno}
\end{align}

Inserting~(\ref{treinta})--(\ref{treintayuno}) in the definition of curvature~(\ref{eFa}),
\begin{equation}
\lambda_{0} F_{a}^{\left( 0 \right)} + \lambda_{1} F_{a}^{\left( 1 \right)} =
\mathrm{d} \left( \lambda_{0} \omega_{a}^{\left( 0 \right)} + \lambda_{1} \omega_{a}^{\left( 1 \right)} \right)
- \frac{1}{2} \eta_{ab} \varepsilon^{bcd}
\left( \lambda_{0} \omega_{c}^{\left( 0 \right)} + \lambda_{1} \omega_{c}^{\left( 1 \right)} \right)
\left( \lambda_{0} \omega_{d}^{\left( 0 \right)} + \lambda_{1} \omega_{d}^{\left( 1 \right)} \right),
\end{equation}
and using the multiplication law~(\ref{e00})--(\ref{e11}), we easily read off the expressions
\begin{align}
F_{a}^{\left( 0 \right)} & =
\mathrm{d} \omega_{a}^{\left( 0 \right)} - \frac{1}{2} \eta_{ab} \varepsilon^{bcd} \left(
\omega_{c}^{\left( 0 \right)} \omega_{d}^{\left( 0 \right)} +
\omega_{c}^{\left( 1 \right)} \omega_{d}^{\left( 1 \right)} \right), \\
F_{a}^{\left( 1 \right)} & =
\mathrm{d} \omega_{a}^{\left( 1 \right)} - \eta_{ab} \varepsilon^{bcd}
\omega_{c}^{\left( 0 \right)} \omega_{d}^{\left( 1 \right)}.
\end{align}

Now we introduce the notation
\begin{align}
\omega_{a}^{\left( 0 \right)} = \omega_{a}, \\
\omega_{a}^{\left( 1 \right)} = \frac{1}{\ell} e_{a},
\end{align}
and find that we can write
\begin{align}
F_{a}^{\left( 0 \right)} & =
\frac{1}{2} \varepsilon_{abc} \left( R^{bc} + \frac{1}{\ell^{2}} e^{b} e^{c} \right), \\
F_{a}^{\left( 1 \right)} & =
\frac{1}{\ell} T_{a},
\end{align}
provided we identify $e^{a}$ with the vielbein, $R^{ab}$ with the Lorentz curvature defined above and $T^{a}$ with the torsion,
$T^{a} = \mathrm{d} e^{a} + \omega^{a}_{\phantom{a}b} e^{b}$.
These last equations correspond to the curvatures of the Lie algebra
$\mathfrak{so} \left( 2, 2 \right)$.
This is a consequence of the fact that the 
$\mathfrak{so} \left( 2, 2 \right)$ algebra can be regarded as a
$\mathbbm{Z}_{2}$-expansion of $\mathfrak{so} \left( 2, 1 \right)$ \cite{Rod06}.

\subsection{Expansion of the Action}

One of the advantages of the $S$-expansion method is that it provides us with an invariant tensor for the $S$-expanded algebra $\G = S \times \g$ in terms of an invariant tensor for \g.

As shown in Ref.~\cite{Iza06b}, a rank 2, symmetric invariant tensor for an $S$-expanded algebra takes the form
\begin{equation}
\left\langle \bm{T}_{\left( A, \alpha \right)} \bm{T}_{\left( B, \beta \right)} \right\rangle_{\G} =
\sigma_{\gamma} K_{\alpha \beta}^{\phantom{\alpha \beta} \gamma} \left\langle \bm{T}_{A} \bm{T}_{B} \right\rangle_{\g},
\end{equation}
where $\sigma_{\gamma}$ are arbitray constants.

Using this result in the $\mathbbm{Z}_{2}$-expansion of the Lorentz algebra $\mathfrak{L}$ we get
\begin{align}
\left\langle \bm{J}_{ab} \bm{J}_{cd} \right\rangle & =
\sigma_{0} \left\langle \bm{J}_{ab} \bm{J}_{cd} \right\rangle_{\mathfrak{L}}, \\
\left\langle \bm{J}_{ab} \bm{P}_{c} \right\rangle & =
\sigma_{1} \left\langle \bm{J}_{ab} \bm{J}_{c} \right\rangle_{\mathfrak{L}}, \\
\left\langle \bm{P}_{a} \bm{P}_{b} \right\rangle & =
\sigma_{0} \left\langle \bm{J}_{a} \bm{J}_{b} \right\rangle_{\mathfrak{L}},
\end{align}
where $\bm{J}_{a} = -\left( 1/2 \right) \varepsilon_{abc} \bm{J}^{bc}$. Inserting~(\ref{einvtenL}) explicitly, we find
\begin{align}
\left\langle \bm{J}_{ab} \bm{J}_{cd} \right\rangle & =
\sigma_{0} \left( \eta_{ad} \eta_{bc} - \eta_{ac} \eta_{bd} \right),
\label{einvtenAdS1} \\
\left\langle \bm{J}_{ab} \bm{P}_{c} \right\rangle & =
\sigma_{1} \varepsilon_{abc}, \\
\left\langle \bm{P}_{a} \bm{P}_{b} \right\rangle & =
\sigma_{0} \eta_{ab}.
\label{einvtenAdS3}
\end{align}
This is perhaps the most important result of this section, so we would like to spell it out: by starting from the Lorentz algebra in $2+1$ dimensions and applying an $S$-expansion with $S = \mathbbm{Z}_{2}$, we find: (i)~the AdS algebra in $2+1$ dimensions, and (ii)~an AdS-invariant, symmetric invariant tensor built from one for the Lorentz algebra. Also note the presence of the arbitrary constants $\sigma_{0}$ and $\sigma_{1}$: they signal the existence of actually \emph{two} independent invariant tensors for the AdS algebra, one built from $\eta_{ab}$ and the other from $\varepsilon_{abc}$.

If we now use the invariant tensor~(\ref{einvtenAdS1})--(\ref{einvtenAdS3}) in the general expression for a CS Lagrangian, we find
\begin{align}
L_{\text{CS}} & = -\frac{1}{2} \sigma_{0} \left( \omega_{\phantom{a}b}^{a} \mathrm{d} \omega_{\phantom{b}a}^{b} +
\frac{2}{3} \omega_{\phantom{a}b}^{a} \omega_{\phantom{b}c}^{b} \omega_{\phantom{c}a}^{c} -
\frac{2}{\ell^{2}} e_{a} T^{a} \right) + \\
& + \frac{\sigma_{1}}{\ell} \varepsilon_{abc} \left[ \left( R^{ab} + \frac{1}{3\ell^{2}} e^{a} e^{b} \right) e^{c} +
\mathrm{d} \left( \frac{1}{8} \omega^{ab} e^{c} \right) \right].
\end{align}
There are two independent terms here: the one proportional to $\sigma_{0}$ provides the `exotic' Lagrangian, while the one proportional to $\sigma_{1}$ contains the EH Lagrangian with a cosmological constant. Note also the presence of the extra torsional term, which ensures the AdS-invariance of the enlarged exotic Lagrangian.

\section{\label{sfin}Comments and Possible Developments}

In this work we have found a relation between two procedures for the expansion of Lie algebras, namely a relation between the so-called `power series expansion method' of Ref.~\cite{deAz02} and the `$S$-expansion Procedure' of Ref.~\cite{Iza06b}. Actually a \emph{dual} formulation of the $S$-expansion method was found, based on the MC forms. It is also shown that the generalization of this procedure permit the construction of $S$-expanded gauge free differential algebras. Finally, as an example of the application of the dual $S$-expansion method, we obtain the $(2+1)$-dimensional CS AdS gravity from the so-called exotic gravity.

Several aspects deserve consideration and many possible developments can be anticipated. A still unsolved problem is to find a relation between five-dimensional CS gravity and general relativity (work in progress).

\begin{acknowledgments}
The authors wish to thank R.~Caroca, J.~Crisóstomo and N.~Merino for enlightening discussions.
F.~I. wishes to thank J.~A.~de~Azcárraga for his warm hospitality at the Universitat de València, where part of this work was done, and for enlightening discussions.
One of the authors (A.~P.) wishes to thank S.~Theisen for his kind hospitality at the MPI für Gravitationsphysik in Golm, where part of this work was done. He is also grateful to the German Academic Exchange Service (DAAD) and the Comisión Nacional de Investigación Científica y Tecnológica (CONICYT), Chile, for financial support.
P.~S. was supported by Fondo Nacional de Desarrollo Científico y Tecnológico (FONDECYT) Grants 1080530 and 1070306 and by the Universidad de Concepción through DIUC Grants 208.011.048 - 1.0.
F.~I. was supported by FONDECYT Grant 11080200, the Vicerrector\'{\i}a de Asuntos Internacionales y Cooperaci\'{o} of the Universitat de València and the Direcci\'{o}n de Perfeccionamiento y Postgrado of the Universidad Cat\'{o}lica de la Santísima Concepción, Chile.
E.~R. was supported by FONDECYT Grant 11080156.
\end{acknowledgments}


\begin{thebibliography}{99}

\bibitem{Hat01}
M.~Hatsuda, M.~Sakaguchi,
\textit{Wess--Zumino Term for the AdS Superstring and Generalized \.{I}n\"{o}n\"{u}--Wigner Contraction}.
Prog. Theor. Phys.~\textbf{109} (2003) 853.
arXiv:~hep-th/0106114.

\bibitem{deAz02}
J.~A.~de~Azc\'{a}rraga, J.~M.~Izquierdo, M.~Pic\'{o}n, O.~Varela,
\textit{Generating Lie and Gauge Free Differential (Super)Algebras by Expanding Maurer--Cartan Forms and Chern--Simons Supergravity}.
Nucl. Phys. B~\textbf{662} (2003) 185.
arXiv:~hep-th/0212347.

\bibitem{deAz04}
J.~A.~de~Azc\'{a}rraga, J.~M.~Izquierdo, M.~Pic\'{o}n, O.~Varela,
\textit{Extensions, Expansions, Lie Algebra Cohomology and Enlarged Superspaces}.
Class. Quant. Grav. \textbf{21} (2004) S1375.
arXiv:~hep-th/0401033.

\bibitem{deAz07}
J.~A.~de~Azc\'{a}rraga, J.~M.~Izquierdo, M.~Pic\'{o}n, O.~Varela,
\textit{Expansions of Algebras and Superalgebras and Some Applications}.
Int.~J. Theor. Phys. \textbf{46} (2007) 2738.
arXiv:~hep-th/0703017.

\bibitem{Iza06b}
F.~Izaurieta, E.~Rodr\'{\i}guez, P.~Salgado,
\textit{Expanding Lie (Super)Algebras through Abelian Semigroups}.
J.~Math. Phys. \textbf{47} (2006) 123512.
arXiv:~hep-th/0606215.

\bibitem{Iza08b}
F.~Izaurieta, E.~Rodr\'{\i}guez, P.~Salgado,
\textit{Construction of Lie Algebras and Invariant Tensors through Abelian Semigroups}.
J.~Phys. Conf. Ser. \textbf{134} (2008) 012005.

\bibitem{Sul77}
D.~Sullivan,
\textit{Infinitesimal Computations in Topology}.
Publ. Math. IH{\'E}S \textbf{47} (1977) 269.
Available at \url{http://www.numdam.org/item?id=PMIHES_1977__47__269_0}.

\bibitem{DAu82}
R.~D'Auria, P.~Fr\'{e},
\textit{Geometric Supergravity in $D=11$ and its Hidden Supergroup}.
Nucl. Phys. B~\textbf{201} (1982) 101.
Erratum-ibid. B~\textbf{206} (1982) 496.

\bibitem{Nak03}
M.~Nakahara,
\textit{Geometry, Topology and Physics}.
Institute of Physics Publishing; 2nd edition (2003).

\bibitem{Wit88}
E.~Witten,
\textit{$(2+1)$-Dimensional Gravity as an Exactly Soluble System}.
Nucl. Phys. B~\textbf{311} (1988) 46.

\bibitem{Rod06}
E.~Rodr\'{\i}guez,
\textit{Formas de Transgresi\'{o}n y Semigrupos Abelianos en Supergravedad}.
Ph.D. Thesis, Universidad de Concepci\'{o}n, Chile (2006).
arXiv:~hep-th/0611032.

\end{thebibliography}
\end{document}